\documentclass[a4paper]{article}
\usepackage{amsthm}
\newtheorem{lemma}{Lemma}
\newtheorem{cor}{Corollary}
\usepackage{graphicx,amssymb,amsmath}

\usepackage[disable]{todonotes}
\usepackage{enumitem}
\newtheorem{claim}{Claim}
\title{Generalized LR-drawings of trees}

\author{Therese Biedl\thanks{David R.~Cheriton School of Computer Science,
        University of Waterloo, Canada {\tt \{biedl,jayson.lynch\}@uwaterloo.ca}}
        \and
\and Giuseppe Liotta%
\thanks{Department of Engineering, University of Perugia, Italy {\tt \{giuseppe.liotta,fabrizio.montecchiani\}@unipg.it}}
\and Jayson Lynch%
\addtocounter{footnote}{-2}
\footnotemark
\and Fabrizio Montecchiani%
\footnotemark
}

\date{}

\begin{document}

\maketitle

\begin{abstract}
The {\em LR-drawing-method} is a method of drawing an ordered rooted binary tree based
on drawing one root-to-leaf path on a vertical line and attaching
recursively obtained drawings of the subtrees on the left and right.
In this paper, we study how to generalize this
drawing-method to trees of higher arity.  We first prove
that (with some careful modifications) the proof of existence
of a special root-to-leaf path transfers to trees of
higher arity.  Then we use such paths to obtain \emph{generalized} LR-drawings
of trees of arbitrary arity.
\end{abstract}

\section{Introduction}\label{se:intro}

Tree-drawing is a very popular topic in the graph drawing literature.
Nearly all tree-drawing methods required that the drawing is {\em planar}
(has no crossing), but there are many other variations, depending on whether
we demand that the drawing is {\em straight-line} (as opposed to permitting
bends in edges) and/or {\em order-preserving} (a given order of edges at each
node is maintained).  For a rooted tree, one also distinguishes by
whether the drawing must be
{\em (strictly) upward} (nodes are (strictly) above their descendants).
%\todo{Do we need to distinguish between strictly and not strictly? Would strictly be enough for our purposes?  TB: For out constructions it doesn't matter, but for the literature review we should list both.}
In all our drawings, we assume that nodes (and also bends, if there are
any) are placed at {\em grid-points}, i.e., points with integer coordinates.
The main objective is to obtain drawings of small {\em area}, measured
via the number of grid-points in the minimum enclosing bounding box of the
drawing.  Sometimes one also considers the {\em width} and
{\em height} of the drawing, measured by the number of columns (respectively rows) that intersect the drawing.
We refer to a survey by Di Battista and Frati \cite{DF14} for many results
up to 2014, and to Chan's recent paper \cite{Chan20} for some development
since.

In 2002, Chan \cite{Chan02} published a tree-drawing paper that became the
inspiration for many follow-up papers.  In particular, he studied rooted
trees and only considered {\em ideal drawings}, i.e., drawings that
respect all four of the above constraints (planar, straight-line,
order-preserving and strictly upward).  His area-results
were superseded by later improvements
\cite{GR03,Bie17-OPTI,Chan20}, but the techniques introduced in \cite{Chan02}
are still widely useful; see e.g.~a recent paper by Frati, Patrignani and
Roselli \cite{FPR20} that uses Chan's recursive approaches to obtain
small straight-line drawings of  outer-planar graphs.

\paragraph{Background and related results.}
One of the methods proposed by Chan \cite{Chan02} is the one that creates
{\em LR-drawings}.  
The idea is to take a root-to-leaf path $\pi$, drawing it
vertically, and attach the left and right subtrees of
the path just below their parent, using a recursively
obtained drawing for the subtree.  See Figure~\ref{fig:LR}, where the nodes of $\pi$ are white (as in all other figures). 
These drawings are defined only for binary trees.

To describe this precisely, we need a few definitions.
Let $T$ be a rooted binary tree that comes with a fixed order of children at each node (we call this an {\em ordered rooted binary tree}).
%\todo[inline]{Shall we use the term \emph{ordered} rooted binary tree? TB: Good idea, done.  Still need to change it in various lemmas}    
A {\em root-to-leaf} path 
$\pi=\langle v_1,v_2,\dots,v_\ell\rangle$  is a path in $T$ where
$v_1$ is the root and $v_\ell$ is a leaf.
A {\em left} subtree of a path $\pi$ is 
a subtree rooted at some child $c$ of a vertex
$v_i\in \pi$ for which $c$ comes before the path-child in the order at $v_i$.  We call this a {\em left subtree at $v_i$} when needing to specify the node of $\pi$.
A {\em right} subtree is defined symmetrically.   

The {\em LR-drawing-method} consists of picking a root-to-leaf path $\pi$,
drawing it vertically,
and attaching (recursively obtained) drawings of the subtrees of 
$T\setminus V(\pi)$ on the left and right side so that the
order is maintained.    Since $T$ is binary, there is only
one such subtree at each $v\in \pi$; we place its drawing in the
rows just below $v$ (after lengthening edges of $\pi$ as needed)
and one unit to the left/right of path $\pi$.

\begin{figure}[ht]
\hspace*{\fill}
\includegraphics[page=1]{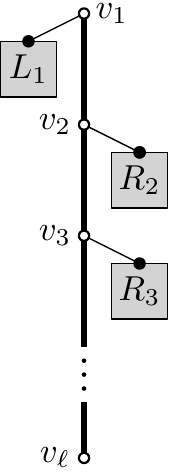}
\hspace*{\fill}
\caption{An LR-drawing where $\pi = \langle v_1, v_2, \ldots v_\ell \rangle$ and $L_i$ and $R_i$ are various subtrees.}
\label{fig:LR}
\end{figure}

%\todo[inline]{Fabrizio: I scaled the figures directly in IPE}

Let $W(n)$ be the maximum (over all binary trees $T$ with $n$ nodes)
of the width of the
drawing, and note that it observes the following recursion: 
%\todo[inline]{isn't it a 2+ .... because we have one unit to the left and one to the right?  Therese: we should define width properly.  Do we count columns or width of the enclosing bounding box?  I've clarified now that we count the columns, and then $+1$ is correct.}
\begin{equation}\label{eq:main}
W(n) \leq 1+\min_\pi \left( 
	\max_{\alpha} W(|\alpha|) + 
	\max_{\beta} W(|\beta|) \right)
\end{equation}
where $\alpha$ and $\beta$ are the maximum left and right subtree,
respectively.
We are therefore interested in picking a path $\pi$ that
has useful bounds on the size of $\alpha$ and $\beta$.
Chan showed that there exists a path such
that $|\alpha|+|\beta|\leq n/2$, for any left and right subtrees $\alpha$ and $\beta$.
%\todo{Added "for any..."}.  
He then improved this to the following:

\begin{lemma}\cite{Chan02}
\label{lem:Chan}
Let $p=0.48$.  Given any binary ordered rooted tree $T$ of size $n$, there exists a root-to-leaf
path $\pi$ such that for any left subtree $\alpha$ and any right subtree $\beta$
of $\pi$, $|\alpha|^p+|\beta|^p \leq (1-\delta)n^p$ for some constant
$\delta>0$.
\end{lemma}

Lemma~\ref{lem:Chan}, together with Equation~\ref{eq:main}, and an inductive argument, implies the existence of LR-drawings of width $O(n^{0.48})$
(and the height is $n$, as it is for all LR-drawings constructed as
described above).  Later on, Frati et al.~\cite{FPR20} showed that 
for some binary trees, a width of $\Omega(n^{0.418})$ is required
in any LR-drawing, and this was improved further to $\Omega(n^{0.429})$
by Chan and Huang \cite{CH20}.  The latter paper also gave another
construction-method that %creates a variant of\todo{Added "a variant of"} LR-drawings, but 
% Therese: no, they are still LR-drawings, it's a variant of the method, not a variant of the drawing
does not follow the above method exactly,
%.  Instead, Chan and Huang pick some root-to-leaf
%path $\pi$ in $T$, but occasionally `twist' path $\pi$, using some (left or
%right) subtree of $\pi$ to complete the vertical line while $\pi$ contains
%some non-vertical segments.  
% Therese: shortened a bit, we say more in the remarks
instead the chosen path $\pi$ may have some non-vertical edges while some left or right subtree of $\pi$ complete the vertically drawn path.
In this way, they can obtain drawings
of width $O(n^{0.437})$.

As should be clear from the above lower bounds, LR-drawings are not
the best tool for small-area ideal tree-drawings, since other papers
can achieve width $O(\log n)$ (or better if the pathwidth is
small) \cite{Bie17-OPTI,GR03} while keeping the height at $n$.  But
LR-drawings have a number of other appealing features:
\begin{itemize}
\item Drawings of disjoint rooted subtrees are ``vertically separated'',
	i.e., if $T_v$ and $T_w$ are two disjoint rooted subtrees, then
	there exists some horizontal line that separates the drawings of
	$T_v$ and $T_w$.  (This can be seen by studying the construction
	at the lowest common ancestor of the two trees.)  As such,
	the drawings are perhaps easier to understand than drawings
	created with other methods (such as~\cite{GR03}) that delay
	the drawing of a subtree until further down, leading to `interleaved'
	drawings of subtrees.
%\item LR-drawings are {\em closed rectangle-of-influence} drawings, i.e., for
%	any edge $(u,v)$, the axis-aligned rectangle spanned by the points
%	assigned to $u,v$ contains no other vertex inside it (see, e.g.,~\cite{DBLP:conf/gd/BiedlBM99,
%%	DBLP:conf/cccg/BiedlLMV16,
%    DBLP:journals/comgeo/LiottaLMW98,
%	DBLP:journals/comgeo/SadasivamZ11}).   
	% Therese: I've removed the paper with Sandor; it's not particularly relevant and I don't like tooting my own horn unnecessarily
%	These
%	types of drawings have been popular   and also have the advantage that (in contrast to arbitrary
%	planar straight-line drawings) one can insert rows and 
%	columns, e.g.~to make space for labels, and planarity is preserved.
%	\todo{Could use reference, but I'm not sure what}
%\todo[inline]{TB: I had put a paragraph about RI-drawings here.  But I now realized that the drawings by Garg and Rusu are also RI-drawings, so this is really not a good excuse for doing LR-drawings.  Removed.}
\item LR-drawings (and in particular Lemma~\ref{lem:Chan}) has been used
	for a number of graph drawing results, including for octagonal
	drawings, orthogonal drawings, and drawings of outer-planar graphs
	\cite{DBF09,Frati07,GR07}.
	\item Last but not least, ``the question on LR-drawings is still interesting 
	and natural, as it is fundamentally about combinatorics of trees,
	or more specifically, decompositions of trees via path separators''
	\cite{CH20}.
\end{itemize}

\paragraph{Contribution.} Our interest in LR-drawings originally came
from the need to generalize Lemma~\ref{lem:Chan} to rooted trees with
higher {\em arity}, i.e., maximum number of children at a node.  (As we
will detail in a separate, forthcoming, paper, such a lemma for ternary
trees can be used to obtain drawings of  outer-1-planar graphs
with smaller area.)      In the process, we discovered that all the
results and applications of LR-drawings seem to be only concerned with
{\em binary trees}.  It is not even clear what exactly an LR-drawing
should be for trees of higher arity, and no area-bounds are known.
Our results in this paper are as follows:

\begin{itemize}
    \item We first show (in Lemma~\ref{lem:paths} in Section~\ref{sec:path})
	 that Lemma~\ref{lem:Chan}  holds for trees of arbitrary arity.  

There does not seem to be an easy way to derive this result from
Chan's result, since it is
not clear how we could modify a rooted tree $T$ into a binary tree without
either increasing the number of nodes or missing a subtree that
may be too big.  For this reason, we re-prove the result from scratch.
The proof is similar in structure to the one by Chan, but we need to be
much more careful in defining the inequalities that hold if we can
extend the path to a subtree.

    \item We then discuss in Section~\ref{sec:generalizedLR} what the
	appropriate generalization of LR-drawings to trees of higher
	arity should be.  We also give a simple construction that shows
	that ideal LR-drawings of area $O(n^2)$  always exist.

    \item In Section~\ref{sec:1bend} and \ref{sec:NonUpward}, we then
	give constructions for generalized LR-drawings that are directly based on
	Lemma~\ref{lem:path} and therefore achieve $O(n^{0.48})$ width.    
	Unfortunately, neither construction gives ideal drawings:
	the first one has one bend per edge, and the second one is not
	upward.  Both constructions can be modified to achieve ideal 
	drawings, but at the expense of increasing the height (possibly more than polynomially).

    \item In Section~\ref{sec:Upward}, we give a construction for
	generalized LR-drawings that are ideal drawings.  The price to
	pay is that the construction is more complicated, and the height (which was linear in the previous constructions)
	increases to $O(n^{1.48})$, meaning that the area is only just
	barely sub-quadratic, namely  $O(n^{1.96})$.
\end{itemize}

We end in Section~\ref{sec:remarks} with open questions.

\section{Choosing a path}
\label{sec:path}

In this section, we show that Lemma~\ref{lem:Chan} can be generalized to
any ordered rooted tree, regardless of its arity.  

\begin{lemma}
\label{lem:path}
\label{lem:paths}
Let $p=0.48$.  Given any ordered rooted tree $T$ of size $n$, there exists a root-to-leaf path $\pi$ in $T$ such that for any left subtree $\alpha$ and any right subtree $\beta$ of $\pi$, $|\alpha|^p+|\beta|^p\leq (1-\delta)n^p$ for some constant $\delta>0$.
\end{lemma}

\begin{proof}
We will iteratively expand path $\pi=\langle v_1,\dots,v_i\rangle$ to get 
closer to a leaf, and let $\alpha_i,\beta_i$ be the largest left/right subtree 
of this path (not considering the subtrees at $v_i$).  
Initially set $v_1$ to be the root.  We maintain the invariant that
$|\alpha_i|^p+|\beta_i|^p\leq (1-\delta)n^p$ for every $i$; this holds
vacuously initially.

Now assume that path $\pi$ up to $v_i$ for some $i\geq 0$ has been chosen.  
Let $S_i^{(1)},\dots,S_i^{(d_i)}$ be the
subtrees at $v_i$, enumerated from left to right.  Call such a subtree
$S_i^{(k)}$ {\em feasible} if we could use its root to extend $\pi$.
Thus $S_i^{(k)}$ is feasible if 
\begin{eqnarray*}
\max\left\{|\alpha_i|,|S_i^{(1)}|,\dots,|S_i^{(k-1)}|\right\}^p \\
+\max\left\{|\beta_i|,|S_i^{(k+1)}|,\dots,|S_i^{(d_i)}|\right\} ^p
& \le & (1-\delta)n^p.
\end{eqnarray*}
For future reference we note that $S_i^{(k)}$ is 
{\em infeasible} if one of the following
three {\em violations} occurs:  
\begin{enumerate}
\item  $|\alpha_i|^p + |S_i^{(\ell)}|^p > (1-\delta)n^p$
for some $\ell>k$.
\item $|S_i^{(h)}|^p + |\beta_i|^p > (1-\delta)n^p$ for some $h<k$.
\item $|S_i^{(h)}|^p + |S_i^{(\ell)}|^p > (1-\delta)n^p$ for some $h<k<\ell$.
\end{enumerate}

\medskip\noindent{\bf Case 1:} Exactly one subtree $S_i^{(k)}$ is feasible.  Then we set $v_{i+1}$ to be the root of $S_i^{(k)}$.  The invariant holds 
by choice of ``feasible''.

\medskip\noindent{\bf Case 2:} At least two subtrees $S_i^{(k)}$, $S_i^{(\ell)}$ (with $k<\ell$) are feasible.  We terminate the construction as follows.

Consider first the subcase where $|S_i^{(k)}|\leq |S_i^{(\ell)}|$.  Set path $\pi$ to be the concatenation of $\langle v_1,\dots,v_i \rangle$  with the leftmost path in $S_i^{(k)}$ down to a leaf.  A left subtree of this path has size at most $\max\{|\alpha_i|,\max_{h<k} |S_i^{(h)}|\}$.  A right subtree of this path up to $v_i$ has size at most $\max\{|\beta_i|,\max_{h> k} |S_i^{(h)}|\}$.  A right subtree of this path below $v_i$ is a subtree of $S_i^{(k)}$ (and hence no bigger than $|S_i^{(k)}|\leq |S_i^{(\ell)}|\leq \max_{h>k} |S_i^{(h)}|$ by assumption).  Since $S_i^{(k)}$ is feasible the invariant holds.

The other subcase, $|S_i^{(\ell)}|\leq |S_i^{(k)}|$ can be handled in a symmetric fashion by extending instead with the rightmost path in $S_i^{(\ell)}$.

\begin{claim} One of the above two cases always applies.  \end{claim}
\begin{proof}
Assume not.  To show that this leads to a contradiction,
we find some subtrees (or collections of subtrees) for which we can lower-bound the size.  This part is significantly more complicated than in Chan's proof because there are now multiple ways in which a subtree might not be feasible, and we must choose our subtrees correspondingly.  Consider Figure~\ref{fig:path1} for an illustration of the following definitions.

\begin{figure}[ht]
\hspace*{\fill}
\includegraphics[page=2]{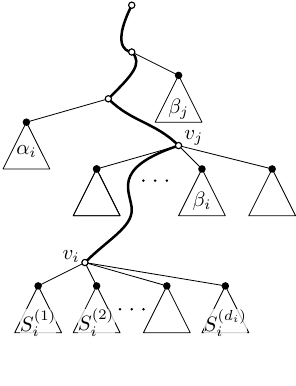}
\hspace*{\fill}
\caption{The situation up to symmetry. } 
\label{fig:path1}
\end{figure}

Suppose that the parent of $\beta_i$'s root, which we denote by $v_j$, is no higher than the parent of $\alpha_i$'s root;  the other case is symmetric.  
We first derive one inequality from $v_j$.
We have $j<i$ by definition of $\beta_i$.  Because we did not terminate the path when extending at $v_j$, Case 2 did not apply at $v_j$.  Therefore the subtree $S_j^{(k)}$ of $v_j$ that contains $v_i$ was the {\em only} feasible subtree at $v_j$.  

Consider Figure~\ref{fig:path2}.
Tree $\beta_i$ is a right subtree at $v_j$, say it was $S_j^{(\ell)}$ with $\ell>k$.  
Let $\mathcal{L}_j$ be the collection of subtrees $S_j^{(1)},\dots,S_j^{(\ell-1)}$.  
We know that $S_j^{(\ell)}$ was infeasible, and study the three possible violations:  \\
(1) $|\alpha_j|^p+|S_j^{(\ell')}|^p > (1-\delta)n^p$ for some $\ell'>\ell$.  But this is impossible since $S_j^{(k)}$ is feasible and $k<\ell$.  \\
(2) $|S_j^{(k')}|^p + |\beta_j|^p>(1-\delta)n^p$ for some $k'<\ell$ (possibly $k=k'$).  In this case, set $B_j:=\beta_j$.  \\
(3) $|S_j^{(k')}|^p+|S_j^{(\ell')}|^p>(1-\delta)n^p$ for some $k'<\ell<\ell'$. In this case, set $B_j:=S_j^{(\ell')}$.  Note that either way $B_j$ is disjoint from $\beta_i=S_j^{(\ell)}$ and a right subtree of one of $v_1,\dots,v_j$, so also disjoint from $\mathcal{L}_j$.  
Subtree $S_j^{(k')}$ that caused the above violation for $S_j^{(\ell)}$ belongs to $\mathcal{L}_j$, and therefore
\begin{equation}
\label{eq:1}
|\mathcal{L}_j|^p + |B_j|^p > (1-\delta)n^p
\end{equation}

\begin{figure}[ht]
\hspace*{\fill}
\includegraphics[page=4]{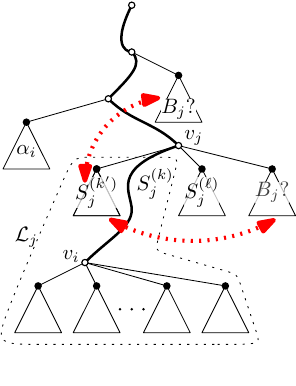}
\hspace*{\fill}
\caption{Close-up on $v_j$.  Red arrows indicates pairs of subtrees that violate feasibility. } 
\label{fig:path2}
\end{figure}

Now we define two subtrees $L_i,R_i$ at $v_i$.
Since the leftmost subtree $S_i^{(1)}$ at $v_i$ is infeasible, but extending into it would not add left subtrees to the path, the infeasibility must be caused by violation (1), i.e., there exists
some $k>1$ such that $|S_i^{(k)}|^p+|\alpha_i|^p > (1-\delta)n^p$.  Set $R_i$ to be this subtree $S_i^{(k)}$, choosing the largest possible index $k$.
We have
\begin{equation}
\label{eq:2}
|R_i|^p + |\alpha_i|^p > (1-\delta)n^p. 
\end{equation}
Symmetrically, since the rightmost subtree $S_i^{(d_i)}$ at $v_i$ is infeasible, there must be
some $h<d_i$ such that $|S_i^{(h)}|^p+|\beta_i|^p > (1-\delta)n^p$.  Set $L_i$ to be this subtree $S_i^{(h)}$, choosing the smallest possible index $h$.  
We have
\begin{equation}
\label{eq:3}
|L_i|^p + |\beta_i|^p > (1-\delta)n^p.
\end{equation}
Note that it is possible $L_i=R_i$ and that both are subsets of $\mathcal{L}_j$.

\medskip\noindent{\bf Case A:} $L_i\neq R_i$.  See Figure~\ref{fig:path3}. In this case
the contradiction is obtained exactly as done by Chan, except by substituting the trees/forests that we have chosen above suitably.
Recall that H\"{o}lder's inequality states that for $p<1$ we have 
$$\sum_a x_a y_a \leq \left(\sum x_a^{1/(1-p)}\right)^{1-p} \left(\sum y_a^{1/p}\right)^p$$  
(in some of our applications below we use $x_a\equiv 1$).
We can derive a contradiction for the value $p=0.48$ (with a sufficiently small $\delta$) by combining Equations (\ref{eq:1}-\ref{eq:3}) as follows:
\begin{eqnarray*}
&& 2.5(1-\delta)n^p  \\
& < & |\alpha_i|^p + |\beta_i|^p + |L_i|^p + |R_i|^p + 0.5 |\mathcal{L}_j|^p + 0.5 |B_j|^p \\
& \leq & |\alpha_i|^p + |\beta_i|^p + 2^{1-p}(|L_i| + |R_i|)^p + 0.5 |\mathcal{L}_j|^p + 0.5 |B_j|^p \\
& \leq & |\alpha_i|^p + |\beta_i|^p + (2^{1-p}+0.5)|\mathcal{L}_j|^p + 0.5 |B_j|^p \\
& \leq & \left(1+1+ (2^{1-p}+0.5)^{1/(1-p)} + 0.5^{1/(1-p)}\right)^{1-p} \\
&& \quad\quad\quad\quad\cdot (|\alpha_i| + |\beta_i| + |\mathcal{L}_j| + |B_j|)^p \\
& < & 2.499n^p
\end{eqnarray*}
since $\alpha_i,\beta_i, \mathcal{L}_i$ and $B_j$ are all disjoint.  
\begin{figure}[ht]
\hspace*{\fill}
\includegraphics[page=3]{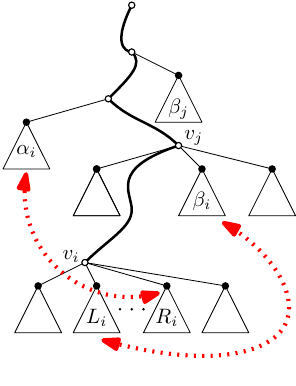}
\hspace*{\fill}
\caption{Close-up on $v_i$, Case (A).}
\label{fig:path3}
\end{figure}

\medskip\noindent{\bf Case B:} $L_i= R_i$.  
In this case we derive one further inequality,
see Figure~\ref{fig:path4}. 
Since $R_i=S_i^{(k)}$ is not feasible, there must exist a violation, and we consider its three possible forms.  (1) $|S_i^{(\ell)}|^p+|\alpha_i|^p > (1-\delta)n^p$ for some $\ell>k$.  But this is impossible since $R_i$ was chosen as the rightmost such violation.  (2) $|S_i^{(h)}|^p + |\beta_i|^p > (1-\delta)n^p$ for some $h<k$.    But this is impossible since $R_i=L_i$ was chosen as the leftmost such violation. 
(3) $|S_i^{(h)}|^p+|S_i^{(\ell)}|^p >(1-\delta)n^p$ for some $h<k<\ell$.  In this case we define $\hat{L}_i:=S_i^{(h)}$ and $\hat{R}_i:= S_i^{(\ell)}$.  Summarizing, $\mathcal{L}_j$ contains the three mutually distinct subtrees $\hat{L}_i,L_i{=}R_i,\hat{R}_i$ and
\begin{equation}
\label{eq:4}
|\hat{L}_i|^p + |\hat{R}_i|^p > (1-\delta)n^p
\end{equation}

\begin{figure}[ht]
\hspace*{\fill}
\includegraphics[page=5]{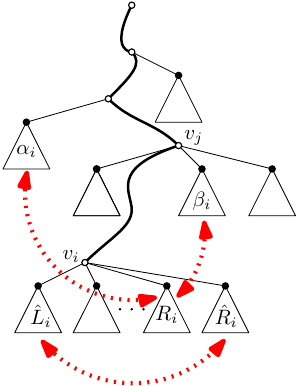}
\hspace*{\fill}
\caption{Close-up on $v_i$, Case (B).}
\label{fig:path4}
\end{figure}

Using again H\"{o}lder's inequality we obtain the desired contradiction:
by combining Equations(\ref{eq:1}-\ref{eq:4}) (and $L_i=R_i$) as follows:
\begin{eqnarray*}
&& 3.5(1-\delta)n^p \\
& < & |\alpha_i|^p + |\beta_i|^p + 2|R_i|^p + |\hat{L}_i|^p + |\hat{R}_i|^p \\
& & \quad\quad + 0.5 |\mathcal{L}_j|^p + 0.5 |B_j|^p \\
& \leq & |\alpha_i|^p +  |\beta_i|^p + (2^{1/(1-p)}+2)^{1-p} (|R_i|+|\hat{L}_i| + |\hat{R}_i|)^p \\
& & \quad\quad + 0.5 |\mathcal{L}_j|^p + 0.5 |B_j|^p \\
& \leq & |\alpha_i|^p + |\beta_i|^p + ((2^{1/(1-p)}+2)^{1-p}+0.5)|\mathcal{L}_j|^p \\
& & \quad\quad + 0.5 |B_j|^p \\
& \leq & (1+1+ ((2^{1/(1-p)}+2)^{1-p}+0.5)^{1/(1-p)} \\
&& \quad\quad+ 0.5^{1/(1-p)} )^{1-p} \cdot \left(|\alpha_i| + |\beta_i| + |\mathcal{L}_j| + |B_j|\right)^p \\
& < & 3.396n^p.
\end{eqnarray*}
So the claim holds.
\end{proof}

So one of Case 1 or Case 2 always applies, and we can continue to expand the path until we terminate it in some application of Case 2.
\end{proof}

\section{Generalized LR-drawings}
\label{sec:generalizedLR}

%\subsection{What should generalized LR-drawings be?}

%\todo[inline]{Added by Beppe and Fabrizio. Feel free to modify.  Therese: I've modified/added a bit.}
%\todo[inline]{Therese: After some discussion with Jayson, this now looks fairly different.  Read carefully, comments welcome.} 

LR-drawings for binary trees were defined via two particular construction operations.  In contrast, we want to define here {\em generalized LR-drawings} 
via the properties that the drawings must satisfy.  
%We are guided by two principles:  The LR-drawings for binary trees must also be generalized LR-drawings, and the lower bounds for the width of LR-drawings should also hold for generalized LR-drawings.   We hence split the conditions on drawings into three groups.     
\todo{TB: I tried to phrase this via guiding principles of generalizing and flexibility, but it came out really awkward. So I decided to simply delete all the stuff about guiding principles instead. jl: looks fine to me, I also gave it a try and couldn't find a good phrasing. The most important part about containing all binary LR is also covered below.}
Let $T$ be an ordered rooted tree, and consider an order-preserving planar drawing $\Gamma$ of $T$.  
%For maximum flexibility, we will only require these two, and 
We call $\Gamma$ a {\em generalized LR-drawing} 
(or {\em GLR-drawing} for short) 
if it (and all induced drawings of rooted subtrees) satisfies 
the following two conditions:
 \begin{enumerate}[label=(P{{\arabic*}})]
 \item \label{it:vert_path}
 ({\em vertical path}): There exists a root-to-leaf path $\pi$ $\langle v_1,\dots,v_\ell\rangle$ that is drawn vertically aligned, with $v_i$ above $v_{i+1}$ for $1\leq 1<\ell$.  
\item \label{it:path_sep}
({\em path-separation}): The column that contains $\pi$ separates the drawings of the left and right subtrees, i.e., for any left or right subtree $T'$, the drawing $\Gamma'$ of $T'$ induced by $\Gamma$ does not use the column containing $\pi$.
%     \todo[]{We need some condition like this.  Alternatively, we could demand that subtree-strips are between $v_{i-1}$ and $v_{i+1}$, but then we have trouble with the non-upward drawings.  } 
 \end{enumerate}
% \todo[inline]{We could avoid mentioning the two guiding principles. We could instead give properties 1 and 2 directly, and observe that they are enough to ensure that an LR-drawing is a GLR-drawing and that  the proofs for the lower bounds transfer to this definition.}
Clearly any LR-drawing for binary trees satisfies these two conditions, so this is indeed a generalization. 
As one can verify by inspecting the proofs, these are the {\em only} two
conditions needed for the lower-bound arguments in \cite{FPR20} and \cite{CH20}. 
We hence have:
 \begin{cor}[based on \cite{CH20}]\label{co:lower-bound}
 For every positive $n$ there exists an $n$-node ordered rooted tree $T$ such that any GLR-drawing of $T$ has width $\Omega(n^{0.429})$.
 \end{cor}
 We note that many existing algorithms for ideal drawings of trees (see e.g.~\cite{Bie17-OPTI,Chan02,GR03}) satisfy the condition on drawing some path $\pi$ vertically.  The real restriction on GLR-drawings is that the vertical path separates the left and right subtrees.
% ;  one verifies that above-listed tree-drawing algorithms all re-use the column of the vertically drawn path for some large subtree that has been ``pushed down''.
% \todo[inline]{We would write  "Observe that, instead, 
The algorithms in [1,2,9] all reuse the column of the vertically-drawn path for some large subtree that has been ``pushed down''.

There are many more properties that  are satisfied by the LR-drawings for binary trees (and so arguably one could have included them in the definition of GLR-drawings, though for maximal flexibility we chose not to do that).  
\todo{FYI: this is what replaced ``guiding principles''}
All LR-drawings of binary trees, and also all drawings that we will create, satisfy the following three properties:
 \begin{enumerate}[label=(P{{\arabic*}})]
 \addtocounter{enumi}{2}
     \item  \label{it:hor_sep}
        ({\em horizontal separation of subtrees}). For any node $v$ of $T$, let $T_v$ be the subtree  rooted at $v$.
%        , and let $\Gamma_v$ be the drawing of $T_v$ induced by $\Gamma$.  
% Therese: we don't really need $\Gamma_v$
There exists a horizontal strip that contains all nodes of $T_v$ and that does not contain any other node of $\Gamma$.
\item  \label{it:group_subtree}
({\em grouping of subtrees at $v_i$}).  For any node $v_i\in \pi$, there exists a horizontal strip that contains $v_i$ and all nodes of all left and right subtrees at %children of 
$v_i$ and that does not contain any other node of~$\Gamma$.  

%(One can actually argue that horizontal separation of subtrees implies grouping of subtrees at $v_i$; we leave the details to the reader.)
\todo[inline]{TB: Jayson had convinced me once that P3 implies P4, but now I no longer believe it.  To be on the safe side, I removed that comment. jl: I don't object to leaving that statement out. I think it follows because $v_{i-1}$ is not in the subtree of $v_i$ and thus can't be in the same band. $v_{i+1}$ is not in the subtree rooted at the other children of $v_i$ so we shouldn't be able to go downwards either.  TB: But some subtree at $v_i$ could ``reach down'' with a long edge and be drawn below $v_{i+1}$ and its subtrees.  (It would look roughly like Fig. 10 (right) except that there is $v_{i+1}$ and some subtree in the bottom left quadrant of $R$.)  Then P3 holds but P4 doesn't.   
}
\item  \label{it:group_left}
({\em grouping of left/right subtrees at $v_i$}).  For any node $v_i\in \pi$, there exists 
a horizontal strip that contains all nodes of all left subtrees at $v_i$ and does not contain any other nodes of $\Gamma$.  
There also exists a horizontal strip that contains all nodes of all right subtrees at $v_i$ and does not contain any other nodes of $\Gamma$.  
 \end{enumerate}
 
% Figure~\ref{fi:TBD} shows a GLR-drawing of a ternary tree \todo[inline]{a figure is needed here}. 

% An LR-drawing of an ordered rooted binary tree $T$ is a particular GLR-drawing of $T$, namely one that is strictly upward and straight-line. For example, the drawing of Figure~\ref{fi:TBD} is not ideal because it is not strictly upward. In other words,  an ideal GLR-drawing of a binary tree is an LR-drawing. 
%It is well understood how to construct GLR-drawings of ordered rooted binary trees in sub-quadratic area that satisfy all three additional criteria (see Section \ref{se:intro}).   

%we focus on constructing  (possibly ideal) GLR-drawings of k-ary trees with good area bounds. By the above definition, we first observe that the width of any GLR-drawing obeys to  Eq.\ref{eq:main}\todo{Please double check.}, where $\alpha$ and $\beta$ are still the maximum left and right subtree with respect to path $\pi$, respectively. Consequently, since the proof of the lower bound by Chan and Huang \cite{CH20} only uses this equation to work \todo{TB: This is NOT correct; the lower bound proof doesn't use this equation at all but argues directly from the vertical paths condition only.  We should delete this paragraph.}, we immediately obtain the following corollary of their work.
 
% \todo[inline]{By reading the proofs of Frati and Chan, it seems to us that they prove the existence of some binary tree where any path $\pi$ is such that Eq. 1 solves to $\Omega(n^{0.429})$. Providing a formal proof might be difficult and lengthy however.}
%\input{lower.tex}
%\subsection{Lower bounds}

Finally, there are three more properties that the LR-drawings of binary trees satisfy, but some of our constructions do not (and as we will argue, we cannot hope to satisfy them and have sub-quadratic area).
 \begin{enumerate}[label=(P{{\arabic*}})]
\addtocounter{enumi}{5}
     \item \label{it:straight} The drawing is straight-line.
     \item \label{it:upward} The drawing is strictly upward.
     \item \label{it:min_distance}
     ({\em minimum-distance}) Bounding boxes of subtree-drawings have the minimum possible distance to the path and to each other.  

Formally, we use $B(T')$ to denote the bounding box of the drawing of rooted subtree $T'$.  We require that if
$T'$ is a left subtree of $\pi$, then the right side of $B(T')$ lies one unit left of $\pi$, and symmetrically for right subtrees.    We also require minimal vertical distances between bounding boxes, most easily expressed by demanding that every row contains a node.
%\todo[inline]{Is `every row contains a node' enough or should we expand that there must be only one unit vertically between bounding boxes? jl: I had thought horizontal-separation + minimum distance implies every row has a node, but now I'm not so sure that closing a gap between bands won't cause any edges to overlap. jl: On second reading, perhaps you want the implication the other way? This clearly won't say anything about columns, so we still need that part, right?}
%Also, if we use the minimum-height horizontal strip for each subtree $T_v$, and also define a (1-row) horizontal strip for each node on $\pi$, then consecutive horizontal strips should be exactly one unit apart.
%Also, if $T_v,T_{v'}$ are two subtrees that consecutive left subtrees at one node $v_i$ of $\pi$, then the bottom side of the bounding box of $T_v$ lies one row above the top side of the bounding box of $T_{v'}$, and symmetrically for right subtrees.    Finally, at each node $v_i\in \pi$ that has a left or right subtree, for one of the subtrees the top of the bounding box lies one unit below $v_i$.
 \end{enumerate}
 \todo[inline]{TB: It occurred to me that we should split P8 into two conditions, one about distances in $x$-direction and one about distances in $y$-direction (call the latter P9).  \\
 For example, the 1-bend construction satisfies the former, but not P9. 
 Also, any drawing that satisfies the P8 has width $O(n^{0.48})$ if we use $\pi$ from Lemma 2, and any drawing that satisfies P9 has height $n$.  It might be worth stating that.
 \\
 I don't want to do such a major change for the CCCG submission this late in the game, but let's do this for the revision (assuming it gets accepted).}

We first show that all eight conditions given above can be satisfied simultaneously if we allow for quadratic area.  The construction is the ``standard
construction''~\cite{Chan20} and hence nearly trivial; we repeat
the details for completeness' sake.

\begin{lemma}
\label{lem:quadratic}
Any $n$-node ordered rooted tree has a GLR-drawing that additionally 
satisfies conditions \ref{it:hor_sep}-\ref{it:min_distance} and has area $O(n^2)$.  Furthermore, the
root is placed in the top-left corner of the bounding box.
%\todo[inline]{sometimes we say node and sometimes vertex.  Probably node everywhere? -Done (until we write more and probably use both again)}
\end{lemma}
\begin{proof}
% On the other hand, a construction to obtain ideal GLR-drawings in $O(n^2)$ area is almost straight-forward, see also Figure \ref{fig:sl}. We aim at recursively computing strictly-upward straight-line GLR-drawings of subtrees of $T$ with the additional invariant that the root occupies the top-left corner of the  bounding box. 
If $T$ consists of a single node, draw such node as an arbitrary point in the plane. Otherwise let $R^{(1)},\dots,R^{(d)}$ be the subtrees 
rooted at the children of the root $v_T$ of $T$, enumerated from \emph{right to left}. 
%\todo[inline]{Therese: Everywhere else we enumerate $S^{(1)},\dots,S^{(d)}$ from left to right, so this is confusing.  Why not use $R^{(1)},\dots,R^{(d)}$ instead, which would be consistent with what we do later (they are right subtrees after all....) Fabrizio: done}
Recursively compute a drawing for each such subtree and combine them as follows. The drawing of $R^{(1)}$ is placed such that the top side of its bounding box $B(R^{(1)})$ is one unit below $v_T$, while its left side is one unit to the right of $v_T$. The drawing of $R^{(i)}$, $1 < i < d$, is placed such that the top side of $B(R^{(i)})$ is one unit below the bottom side of $B(R^{(i-1)})$, while its left side is again one unit to the right of $v_T$. The drawing of $R^{(d)}$ is drawn such that its root is vertically aligned with $v_T$ and the top side of $B(R^{(d)})$ is one unit below the bottom side of $B(R^{(d-1)})$. (This corresponds to choosing $\pi$ as the leftmost path of $T$.) It is easy to verify that each edge can be drawn as a straight-line segment, and that the resulting drawing is a strictly-upward GLR-drawing 
%that maintains the desired invariant.   % TB: no more invariants
that satisfies all conditions.
Moreover, the construction guarantees that both the width and the height are at most $n$.  
\end{proof}
 \begin{figure}[ht]
\centering
\includegraphics[page=1]{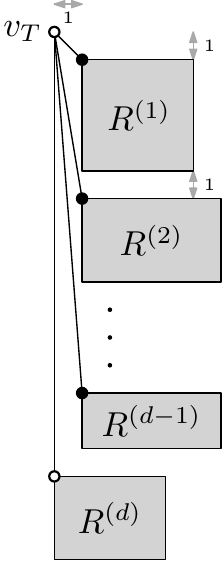}
\caption{Construction an ideal GLR-drawing of a $d$-ary tree in $O(n^2)$ area.\label{fig:sl}
%\todo[inline]{TB: This figure is extremely similar to the left one of Figure~\ref{fig:Upward}.  If space is an issue, delete this figure.}
}
\end{figure}

Unfortunately, it turns out that we cannot hope for smaller area if
we want to satisfy all conditions.

%\todo[inline]{We added "For every positive $n$", is it true or do we have some lower bound for $n$?  Therese: not really needed (the $\Omega$ includes that $n$ is sufficiently big) but doesn't hurt.}
\begin{lemma}
\label{lem:allConditions}
For every positive integer $k$ there exists an  ordered rooted tree $T$ with $n= 6k -1$ nodes and arity 4 such that any GLR-drawing of $T$ that also satisfies conditions
\ref{it:hor_sep}-\ref{it:min_distance}   for all subtrees has area $\Omega(n^2)$.
\end{lemma}
\begin{proof}
Tree $T$ consists of root $r$ with four children $\ell,v_1,v_1',\ell'$,
where $\ell,\ell'$ are leaves while $v_1$ and $v_1'$ are each roots of trees
of height $k$.
Specifically, for $1\leq i<k$, nodes $v_i$ and $v_i'$ each have three children: one leaf, node $v_{i+1}$ resp.~$v_{i+1}'$, and another leaf.  See Figure~\ref{fig:quadratic} for the case when $k=4$.

\begin{figure}[ht]
\hspace*{\fill}
\includegraphics[page=1]{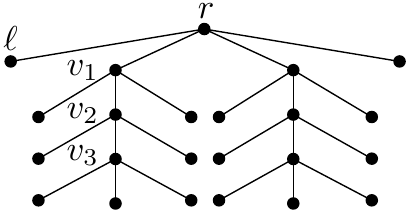}
\hspace*{\fill}
\includegraphics[page=2]{quadratic.pdf}
\hspace*{\fill}
\caption{Construction of a family of trees that require quadratic area in any GLR-drawing that satisfies all conditions \ref{it:hor_sep}-\ref{it:min_distance}.}
\label{fig:quadratic}
\end{figure}

Fix an arbitrary GLR-drawing of $T$ that satisfies criteria 
\ref{it:hor_sep}-\ref{it:min_distance}.  
Up to symmetry, we may assume that the vertically drawn path $\pi$ uses $v_1'$ or $\ell'$, so $\ell$ and $T_{v_1}$ are both left subtrees.  By the minimum-distance condition, 
and since $T_\ell$ only consists of $\ell$,  the location of $\ell$ is one unit left of path $\pi$.  (It may be one or two units below $r$, depending on the location of $\ell'$, but our proof will not use this.)
Now consider the bounding box $B(T_{v_1})$ of the drawing of $T_{v_1}$.
Again by the minimum-distance condition, its right side must be one unit left of path $\pi$.  Since left subtrees at root $r$ are grouped, distances are minimum and the drawing is ordered, the top side of $B(T_{v_1})$ must be one unit below $\ell$.  Since the drawing is strictly-upward, node $v_1$ is in the top row of $B(T_{v_1})$.  Since $(r,v_1)$ is drawn with a straight-line segment, this requires $v_1$ to be on the top right corner of $B(T_{v_1})$, otherwise the drawing would not be order-preserving at $r$ or $(r,v_1)$ would overlap $\ell$.

So we now know that $T_{v_1}$ is drawn with its root in the top right corner.  
It follows that the path $\pi'$ used for drawing $T_{v_1}$ must use the
right child of $v_1$, i.e., goes to a leaf while $T_{v_2}$ is a left subtree
of $\pi'$.  But notice that now the situation is repeated: at $v_1$, there exists a leaf, then the subtree $T_{v_2}$, and both are left subtrees of the vertically drawn path $\pi'$.  As above one argues that hence $v_2$ is drawn on the top right corner of the bounding box $B(T_{v_2})$ of $T_{v_2}$, and one unit left of $v_1$.  Repeating the argument, each $v_i$ is drawn one unit further left that $v_{i-1}$.  Therefore $B(T_{v_1})$ has width at least $k\in \Omega(n)$.  We also know that $B(T_{v_1})$ has height at least $k$ since its tree has height $k$ and is drawn strictly upward.  So the area is $\Omega(n^2)$.
\end{proof}

So in our constructions, we relax some of the conditions \ref{it:straight}-\ref{it:min_distance},
and show that then we can achieve subquadratic-area GLR-drawings.

\subsection{Upward 1-bend GLR-drawings}
\label{sec:1bend}

Let $T$ be an ordered rooted tree, and let $\pi=\langle v_1,v_2,\dots,v_\ell \rangle$ be a root-to-leaf path of $T$. We give a simple recursive construction to compute a strictly-upward generalized LR-drawing of $T$ by using at most one bend per edge. Refer to Figure~\ref{fig:1bend} for an illustration. 

\begin{figure}[ht]
\centering
\includegraphics[page=2]{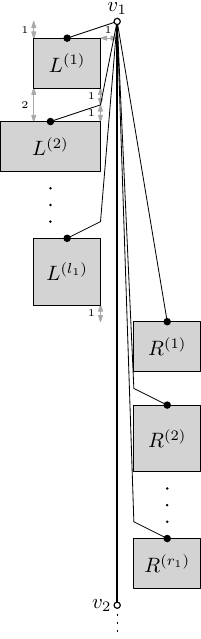}
\caption{Construction for 1-bend generalized LR-drawings. 
}
\label{fig:1bend}
\end{figure}

Assume $v_1$ is placed at an arbitrary point of the plane. For ease of notation, let $L^{(1)},\dots,L^{(l_1)}$ be the left  subtrees rooted at $v_1$, enumerated from left to right. Recursively compute a drawing for each $L^{(i)}$, $1 \le i \le l_1$, and denote by $B(L^{(i)})$ the corresponding bounding box. Place the drawing of $L^{(1)}$ such that the top side of $B(L^{(1)})$ is one unit below $v_1$ and so that the right side of $B(L^{(1)})$ is one unit to the left of $v_1$. Similarly, for each $1 < i \le l_1$, place the drawing of $L^{(i)}$ such that that the top side of $B(L^{(i)})$ is two units below the bottom side of $B(L^{(i-1)})$ and so that the right side of $B(L^{(i)})$ is one unit to the left of $v_1$. Let $R^{(1)},\dots,R^{(r_1)}$ be the right subtrees at $v_1$, enumerated from {\em right to left}. 
Apply a symmetric construction as for the left subtrees and move them down such that the top side of $B(R^{(1)})$ is placed one unit below the bottom side of $B(L^{(l_1)})$. Now place $v_2$ vertically aligned with $v_1$ and one unit below the bottom side of $B(R^{(r_1)})$. Concerning the edges, observe that the edge connecting $v_1$ to the root of $L^{(1)}$ can be drawn with a straight-line segment without crossings, whereas the other edges that connect $v_1$ to the root of each subtree $L^{(i)}$, with $i>1$, can instead be drawn with precisely one bend placed one unit above the top side of $B(L^{(i)})$ and one unit to the left of $v_1$. The edges that connect $v_1$ to the roots of the subtrees $R^{(i)}$ are drawn symmetrically. 
By repeating the construction for each $v_i$, with $1 < i \le \ell$, we conclude the drawing. 

Every row contains a node or a bend, so the height is $O(n)$.
The width
obeys Eq.~\ref{eq:main}. When path $\pi$ is chosen as prescribed by Lemma~\ref{lem:path}, Chan \cite{Chan02} proved that Eq.~\ref{eq:main} solves to $O(n^{0.48})$. The next lemma follows.

\begin{lemma}\label{le:1bend}
Any ordered rooted tree of size $n$ admits a strictly-upward GLR-drawing with at most one bend per edge, whose width is $O(n^{0.48})$ and whose height is $O(n)$.
\end{lemma}

Note that the above GLR-drawings 
can be vertically  stretched so to become straight-line and hence an ideal GLR-drawing. Namely, for each edge $(u,v)$ drawn with one bend, it suffices to insert sufficiently many rows above the bend point so to guarantee a direct line of sight between $u$ and $v$. This is always possible because each bend point is such that no other node or bend is placed with the same $y$-coordinate. While the above transformation does not change the width of the drawing, it may produce a height that is not polynomial in $n$. 
Also, it does not satisfy \ref{it:min_distance}.  (The drawing of Lemma~\ref{le:1bend} does not satisfy \ref{it:min_distance} either because of the rows for the bends, but at  least it comes close.)

\subsection{Non-upward straight-line LR-drawings}
\label{sec:NonUpward}

In this section, we show how we can avoid using bends.  Thus we create a GLR-drawing that is
straight-line, and in fact satisfies all 
of \ref{it:hor_sep}-\ref{it:min_distance} except that it is not upward.
The crucial idea is to give two drawing-algorithms to create 
different types of GLR-drawings.
\begin{itemize}
\item In a type-I drawing, the root is located in the top row (with no
	restriction on the column).
\item In a type-II drawing, the root is located in the leftmost or
	rightmost column, with no  node above it in the same column. 
	We will
	use type-II$\ell$ and type-II$r$ to specify whether the root
	is left or right.
\end{itemize}

\begin{lemma}
Let $p=0.48$.  Given any ordered rooted tree $T$ of size $n$, there exist 
\begin{itemize}
\item a straight-line GLR-drawing of type I that has width at most $cn^p-1$ (for some constant $c>0$),
\item a straight-line GLR-drawing of type II$\ell$ that has width at most $cn^p$ (for the same constant $c$), and
\item a straight-line GLR-drawing of type II$r$ that has width at most $cn^p$ (for the same constant $c$).
\end{itemize}
Furthermore, all drawings have height $n$.
\end{lemma}
\begin{proof}
If $T$ consists of a single node then the claim holds trivially.
Otherwise, pick a path $\pi=\langle v_1,\dots,v_\ell \rangle$ with Lemma~\ref{lem:path}.

We first explain how to create type-I drawings, which is very similar to 
Section~\ref{sec:1bend} except that we use type-II drawings to avoid bends.
Assume $v_1$ is placed at an arbitrary point of the plane. 
Let $L^{(1)},\dots,L^{(l_1)}$  and $R^{(1)},\dots,R^{(r_1)}$ be as in Section~\ref{sec:1bend}. 
Recursively compute drawings as follows:
\begin{itemize}
    \item  
a type-I drawing for $L^{(1)}$ and $R^{(1)}$, \item a type-II$r$ drawing for each $L^{(i)}$, $2 \le i \le l_1$, and 
\item a type-II$\ell$ drawing for each $R^{(i)}$, $2\le i \le r_1$.
\end{itemize}

Place the drawings as in Section~\ref{sec:1bend}, except leave only one unit vertical distance between the bounding boxes.
As before, the edges from $v_1$ to the roots of $L^{(1)}$ and $R^{(1)}$ can be drawn straight-line without crossing.
The edges to all other children can now also be drawn straight-line since those children are in an adjacent column to $v_1$.

Any left subtree $\alpha$ uses at most $c|\alpha|^p$ columns and any right subtree $\beta$ uses at most $c|\beta|^p$ columns by induction, so by Lemma~\ref{lem:paths} the width is at most
$$ c|\alpha|^p+c|\beta|^p+1  \leq  c(1-\delta) n^p  +1\leq  cn^p-1$$
for the constant $\delta>0$ from Lemma~\ref{lem:paths} and assuming $c$ is sufficiently large.

\begin{figure}[ht]
\hspace*{\fill}
\includegraphics[page=1]{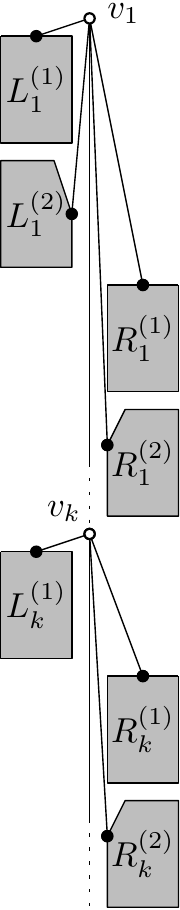}
\hspace*{\fill}
\includegraphics[page=2]{NonUpward.pdf}
\hspace*{\fill}
\caption{Constructions of straight-line drawings with linear height.  (Left) Type-I drawings.  (Right) Type-II$\ell$ drawings. 
}
\label{fig:NonUpward}
\end{figure}

Now we turn towards type-II drawings and only explain how to create a
type-II$\ell$ drawing; the other type is symmetric.    Let $\pi=\langle v_1,\dots,v_\ell \rangle$, and let $k\geq 1$ be the minimal index such that $v_k$ has a left subtree.  (If there is no such $v_k$ then the type-I drawing is in fact a type-II$\ell$ drawing.)
We draw $v_1,\dots,v_{k-1}$ as we did for type-I drawings; since they do not have left subtrees this places $v_1,\dots,v_{k-1}$ in the leftmost column.  
At $v_k$, we proceed as follows:
\begin{itemize}
\item As in Section~\ref{sec:1bend}, let $L^{(1)},\dots,L^{(l_k)}$ and $R^{(1)},\dots,R^{(r_k)}$ be the left and right subtrees at $v_k$.
\item We use a type-II$\ell$ drawing for $R^{(1)},\dots,R^{(r_k)}$, and denote by $B(R^{(i)})$ the corresponding bounding box. Place the drawing of $R^{(1)}$ such that the top left of $B(R^{(1)})$ is one unit below the bottom left corner of the bottommost subtree of $v_{k-1}$.  (If $k=1$, then place $R^{(1)}$ arbitrarily.)
For $i=2,\dots,r_k$, place the drawing of $R^{(i)}$ such that the top left corner of $B(R^{(i)})$ is one unit below the bottom left corner of $B(R^{(i-1)})$.

\item  Place $v_k$ in the next row below.  If $k>1$, place $v_k$ vertically below $v_{k-1}$.  If $k=1$ and $r_1>0$, place $v_k$ such that it is one unit to the left 
of $B(R^{(1)})$.  If $k=1$ and there was no right subtree, then place $v_k$ arbitrarily.

\item Let $S(\pi)$ be the subtree rooted at $v_{k+1}$ (thus containing the rest of $\pi$).   Use a type-I drawing for $S(\pi)$, and place it in the rows below $v_k$, with the left side of $B(S(\pi))$ one unit to the right of $v_k$.

\item We use a type-II$\ell$ drawing 
for $L^{(1)},\dots,L^{(l_k)}$.  We know that  $l_k>0$ by choice of $k$.
If $l_k>1$, then place $L^{(l_k)}$ such that the top left corner of $B(L^{(l_k))})$ is one unit below the bottom left corner of $B(S(\pi))$. For $i$ from $l_k-1$ down to 2, place the drawing of $L^{(i)}$ such that the top left corner of $B(L^{(i)})$ is one unit below the bottom left corner of $B(L^{(i+1)})$.  Finally, place the drawing of $L^{(1)}$
in the next rows such that the top left corner of $B(L^{(1)})$ is exactly below $v_k$.  
\end{itemize}

Thus, as in \cite{CH20}, the vertically drawn path is {\em not} the path $\pi$ that we started out with, instead it is
$v_1,\dots,v_k$ plus the vertically-drawn path of $L^{(1)}$. 
But still we obtain a GLR-drawing.
%One easily verifies that this is a generalized LR-drawing.
To prove that this drawing has the appropriate width, we need an observation.
\begin{claim}
Let $T'$ be a left or right subtree of path $\pi$ chosen with Lemma~\ref{lem:paths}.  Then $T'$ has size at most
$(1-\delta)^{1/p} n$, for the constant $\delta>0$ from Lemma~\ref{lem:paths}.
\end{claim}
\begin{proof}
This follows directly from the bound in the lemma since necessarily $|T'|^p\leq (1-\delta)n^p$.
\end{proof}

Therefore, at any subtree other than $S(\pi)$, the width is by induction
at most $1+c(1-\delta)n^p \leq cn^p$ for sufficiently large $c$.  At subtree
$S(\pi)$, the recursively obtained type-I drawing has width at most $cn^p-1$,
and so the width again is at most $cn^p$ as desired.

In all cases, we never insert empty rows between drawings, so every row
contains exactly one node 
and the height is $n$ as desired.    
\end{proof}

As in Section~\ref{sec:1bend}, we can stretch the drawing vertically to make it upward, by moving all subtrees at $v_k$ downward and leaving a sufficiently large gap between $B(R^{(r_k)})$  and $B(S(\pi)))$ so that edge $(v_k,v_{k+1})$ can be routed straight-line.  (Details are left to the reader.)  Again the height may not be polynomial and condition \ref{it:min_distance} no longer holds.
\todo{New minor remark; could be omitted if space is tight.}

\subsection{Upward straight-line LR-drawings}
\label{sec:Upward}

As shown in the previous sections, we can achieve width $O(n^{0.48})$,
but the drawings are either not straight-line, or not upward, or have
large (possibly super-polynomial) height.  In this section, we show that with a different
construction, we can bound the height to be $O(n^{1.48})$ in a straight-line, upward GLR-drawing of width $O(n^{0.48})$.  The area
hence is $O(n^{1.96})$, just barely under the trivial $O(n^2)$ bound.

The idea for this is to follow a {\em different} approach of Chan
(`Method 4') for tree-drawing; here we occasionally double the
height used for some subtrees, but this happens rarely enough that
overall the height can still be bounded.  Chan's Method 4 does not produce GLR-drawings (he lets the largest subtree re-use the column of the vertical path)
but we can modify the approach at the cost of increasing
the width by one unit.  To compensate for this we use a type-I
drawing for the largest subtree, so that the overall width does not increase
too much.  

So again we have drawing-types.  One of them is exactly the
type-I drawing used in the previous section.  The other one,
which we call type-III drawing, has the root located in the  top left or top right corner; we 
	use type-III$\ell$ and type-III$r$ to specify whether the root
	is left or right.

\begin{lemma}
Let $p=0.48$.  Given any ordered rooted tree $T$ of size $n$, there exist 
\begin{itemize}
\item a straight-line upward GLR-drawing of type I that has width at most $cn^p-1$ (for some constant $c>0$),
\item a straight-line upward GLR-drawing of type III$\ell$ that has width at most $cn^p$ (for the same constant $c$), and
\item a straight-line upward GLR-drawing of type III$r$ that has width at most $cn^p$ (for the same constant $c$).
\end{itemize}
Furthermore, all drawings have height at most $2n^{1+p}$.
\todo{The height doesn't need the constant $c$; we can explicitly write 2 here.  If anyone is re-reading this part, focus on whether all those equations got changed correctly. jl: I looked through and believe the equations were correctly updated. The phrasing accommodates one and a half rows' strikes me as a bit strange but makes sense in the accounting.}
\end{lemma}
\begin{proof}
 
1
If $T$ consists of a single node then the claim holds trivially.
Otherwise, pick a path $\pi$ with Lemma~\ref{lem:path}, so that $|\alpha|^p+|\beta|^p\leq (1-\delta)n^p$ for any left and right subtrees $\alpha,\beta$ of the path.    The creation of a type-I drawing is exactly as in the previous section, except that we use type-III drawings in place of type-II drawings so that we have an upward drawing.  Using $H(\cdot)$ to denote the height, we have
$H(n)\leq \sum_{i=1}^d H(n_i)+1$ 
where $d$ is the number of subtrees and $n_i$ is the size of the $i$th subtree.
Since $n_i\leq n$ and $\sum_i n_i=n-1$ 
%and $n^{1+p}$ is a concave function, 
we have $$\sum_{i=1}^d H(n_i) +1\leq \sum_{i=1}^d 2n_in^p+ 2n^p 
\leq 2n^p(\sum_{i=1}^d n_i + 1)=2n^{1+p}.$$ 
To construct type-III drawings, we proceed much as in Chan \cite{Chan02}, Method 4.  Fix $A= n/2^{1/p}>0.23n$.  (The value of $A$ is different from Chan's, but its use is nearly the same.)    We will completely disregard path $\pi$ and instead pick one subtree of the root based on $A$.  To simplify notations, let $S_1,\dots,S_d$ be the subtrees at the root, enumerated from left to right.  We only explain how to construct a type-III$\ell$ drawing; constructing type-III$r$ drawings is symmetric.  We have two cases:

\medskip\noindent{\bf Case 1:}  Either $d\leq 2$ or the subtrees $S_2,\dots,S_{d-1}$ all have size at most $n-A$.    In this case, recursively construct a type-III$\ell$ drawing for $S_1,\dots,S_{d-1}$ and a type-I drawing for $S_d$.
Combine these drawings with the standard method that was used already in 
Lemma~\ref{lem:quadratic}, see also Figure~\ref{fig:Upward}.  Clearly this is a planar order-preserving straight-line upward drawing, and its height is $1+\sum_i H(n_i) \leq 2n^{1+p}$ with the same analysis as for type-I drawings.
The width $W(n)$ is at most
$1+(cn^p-1)=cn^p$ at $S_d$, and at most $W(n-1)\leq cn^p$ at $S_1$.  At
any other subtree $S_i$, the size is at most $n-A<0.77n$, and the width is at most
%$$1+W(n-A) \leq 1+c(n(1-1/2^{1/p}))^p \leq  cn^p$$
$1+W(0.77n)\leq 1+c(0.77)^pn^p \leq cn^p$,
assuming $c$ is sufficiently large.

\begin{figure}[ht]
\hspace*{\fill}
\includegraphics[page=2]{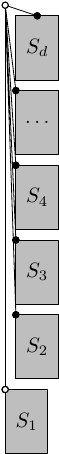}
\hspace*{\fill}
\includegraphics[page=3]{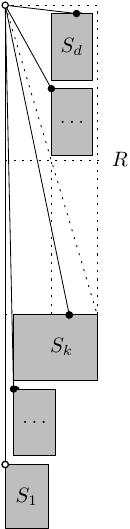}
\hspace*{\fill}
\caption{Constructions of type-III$\ell$ straight-line upward drawings with subquadratic area.  (Left) Case 1. (Right) Case 2.
}
\label{fig:Upward}
\end{figure}

\medskip\noindent{\bf Case 2:}  $d\geq 3$ and at least one subtree $S_k$ with $1< k< d$ has size $n-A$ or more.    
In this case, use a type-I drawing for $S_k$ and $S_d$, and a type-III$\ell$ drawing for all other subtrees.  

We explain how to build the drawing in Figure~\ref{fig:Upward} bottom-up.  Place the drawing of $S_1$ arbitrarily.  Place $S_2,\dots,S_k$ on top of this, with one unit between their bounding boxes, such that the left sides of their bounding boxes are one unit to the right of the root of $S_1$.   Now place an imaginary $W'\times H'$ rectangle $R$ with its left side aligned with the left side of $B(S_1)$ and its bottom side coinciding with the top side of $B(S_k)$.
Here $W'=\max\{|S_k|^p,\max_{i>k} \{2W(|S_i|)\}$, while $H'=3+2\sum_{i>k}H(|S_i|)$.
In particular, the top right quadrant of $R$ is big enough to accommodate the drawings of $S_{k-1},\dots,S_1$, plus one and a half rows. 
We place the root at the top left corner of $R$.  We place the drawings of $S_{k-1},\dots,S_1$ (in this order) below the root and in the top right quadrant of $R$, with the left sides of their bounding boxes aligned and unit vertical distance between boxes.
Note that this does {\em not} use the center-point of $R$ (which is not a grid point due to the odd height of $R$).
The root of $S_k$ is somewhere along the bottom of $R$, hence 
the edge from it to the root of $T$ does not enter the top right quadrant
of $R$ and the drawing is planar.  Clearly it is a GLR-drawing and also strictly upward. In fact, conditions \ref{it:hor_sep}-\ref{it:upward} are all satisfied (but \ref{it:min_distance} is not).

We first analyze the width. At $S_k$,  the width is at most $1+(cn^p-1)\leq cn^p$ by induction.
At any other subtree $S_i$, the size is $n_i\leq A=n/2^{1/p}$
and the width is at most 
$$\max\{1+W(|S_i|),2W(|S_i|) \} \leq 2\cdot c\left(\tfrac{n}{2^{1/p}}\right)^p  = cn^p$$
as desired.  As for the height, $S_k$ contributes at most
$2n_k^{1+p}$ rows
and any other
subtree $S_i$ contributes at the most $2\cdot 2n_i^{1+p}$ rows.  We need two further rows for $R$ (the bottom row was already counted).  Since $n_i\leq A=n/2^{1/p}$ for $i\neq k$ and $\sum_{i} n_i=n-1$, the height is at most
$$2+2n^p n_k + \sum_{i\neq k} 4\left(\tfrac{n}{2^{1/p}}\right)^p n_i  \leq 2n^p + 2n^p \sum_{i=1}^d n_i  = 2n^{1+p}$$
as desired.
\end{proof}

\section{Remarks}
\label{sec:remarks}

In this paper, we studied how to generalize the concept of LR-drawings 
that was previously designed for binary trees~\cite{Chan02,CH20,FPR20}
to trees of higher arity.  To this end, we first generalized a lemma by Chan
about paths for which $|\alpha|^p+|\beta|^p\leq (1-\delta)n^p$ for
constant $p=0.48$, $\delta>0$ and any left and right subtree $\alpha,\beta$.  Then we
explained how to use this path to construct generalized LR-drawings of width $O(n^{0.48})$ and subquadratic area, 
both with and without the
restriction on the drawing being straight-line and/or upward. We conclude the paper by listing some open problems

\begin{itemize}
    \item 
The most natural open problem is to close the gap on the width of GLR-drawings.
Frati et al.~showed that width $\Omega(n^{0.418})$ is sometimes required \cite{FPR20},
and Chan and Huang improved this to $\Omega(n^{0.428})$ \cite{CH20}.  These lower bounds
were for binary trees;
%so of course they also hold for arbitrary arity.  But
could they perhaps be strengthened if we allow higher arity?
Using ternary trees, one can immediately reduce the size of the lower-bound tree $T_h$ of 
\cite{CH20,FPR20}, by $2^h-1$ (contract every second edge of path $\pi$) without affecting the validity of the lower-bound proof.  Unfortunately, this improves the lower bound only by a lower-order term.

\item We showed that in a GLR-drawing where all additional conditions \ref{it:hor_sep}-\ref{it:min_distance} are satisfied, the width must be $\Omega(n)$.  
Is there an intermediate lower bound that shows up when requiring other subsets of these properties?  We are especially curious about removing condition \ref{it:group_left} (`grouping of left/right subtrees'), which seems very artificial but is crucially required in the proof of Lemma~\ref{lem:allConditions}.
\item
Chan and Huang improved the width of LR-drawings of binary trees \cite{CH20}.  The main idea is that rather than drawing one chosen path $\pi$ as a straight-line, they add an `$i$-twist' to the drawing of path 
$\pi$, using $2^i$ non-vertical edges for $\pi$ while the corresponding other subtrees at these edges use vertical lines.
With this, they can achieve an LR-drawing of width $O(n^{0.438})$ (and even smaller with further improvements).    

It would be interesting to see whether this approach could be generalized to trees of higher arity.  
We cannot generalize the algorithm directly, because 
the subtrees that use vertical lines are defined via a size-property.  In trees of higher arity these  subtrees may well have common parents on $\pi$,
making it impossible to use distinct vertical lines for them, as is necessary in the construction.

\item Our construction of ideal GLR-drawings achieves subquadratic area, but barely.  Can the area be improved?  
\end{itemize}

\small
\bibliographystyle{abbrv}
\bibliography{papers}

\end{document}